\documentclass[fleqn,10pt]{wlscirep}
\usepackage[utf8]{inputenc}
\usepackage{amsfonts}
\usepackage{amsmath}
\usepackage{algorithm}
\usepackage{algpseudocode}
\usepackage{color}
\usepackage{graphicx}
\usepackage{amsthm}
\usepackage{comment}
\usepackage{multirow}
\usepackage{hyperref}
\usepackage{titlesec}
\usepackage{xr-hyper}

\newtheorem{definition}{Definition}
\newtheorem{theorem}{Theorem}

\newtheorem{proposition}{Proposition}
\newtheorem{corollary}{Corollary}

\usepackage[T1]{fontenc}
\title{A User Centric Group Authentication Scheme for Secure Communication}

\author[1,*]{Oylum Gerenli}
\author[2]{Gunes~Karabulut-Kurt}
\author[1]{Enver Ozdemir}

\affil[1]{Istanbul Technical University, Informatics Institute, Istanbul, Turkey}
\affil[2]{Polytechnique Montr\'eal, Department of Electrical Engineering, Montr\'eal, Canada}

\affil[*]{gerenli23@itu.edu.tr}

\begin{abstract}
Group Authentication Schemes (GAS) are methodologies developed to verify the membership of multiple users simultaneously. These schemes enable the concurrent authentication of several users while eliminating the need for a certification authority. Numerous GAS methods have been explored in the literature, and they can be classified into three distinct generations based on their foundational mathematical principles.  First-generation  GASs rely on polynomial interpolation and the multiplicative subgroup of a finite field. Second-generation GASs also employ polynomial interpolation, but they distinguish themselves by incorporating elliptic curves over finite fields. While third-generation GASs present a promising solution for scalable environments, they demonstrate a limitation in certain applications. Such applications typically require the identification of users participating in the authentication process. In the third-generation GAS, users are able to verify their credentials while maintaining anonymity. However, there are various applications where the identification of participating users is necessary. In this study, we propose an improved version of third-generation GAS, utilizing inner product spaces and polynomial interpolation to resolve this limitation.  We address the issue of preventing malicious actions by legitimate group members. The current third-generation scheme allows members to share group credentials, which can jeopardize group confidentiality. Our proposed scheme mitigates this risk by eliminating the ability of individual users to distribute credentials. However, a potential limitation of our scheme is its reliance on a central authority for authentication in certain scenarios. 
\end{abstract}
\begin{document}

\flushbottom
\maketitle
\thispagestyle{empty}

%\noindent Please note: Abbreviations should be introduced at the first mention in the main text – no abbreviations lists. Suggested structure of main text (not enforced) is provided below.

\section*{Introduction}

The Internet of Things (IoT) is a network that utilizes the internet to interconnect physical and virtual devices, enabling intelligent decision-making. IoT has numerous potential for real-time applications, such as environmental monitoring, healthcare service, inventory and production management, food supply chain (FSC), transportation, workplace, and home support \cite{da2014internet}. However, IoT security has become a pressing concern due to the rapid expansion of IoT systems, which requires the protection of both networks and hardware components \cite{frustaci2017evaluating}. Among the various security aspects, authentication is particularly vital for the efficient operation of digital communication networks, as it ensures that the devices participating in the network are legitimate. Standard authentication and key establishment mechanisms may not be suitable for devices communicating over wireless mediums, particularly  IoT devices. Due to the resource constraints of IoT devices, implementing public key algorithm-based methods \cite{chaudhary2018lattice} for authentication is not ideal, as such methods involve operations with large numbers. Additionally, current and next-generation wireless communication systems (5G and beyond) \cite{chettri2019comprehensive,nguyen20216g} will connect a vast array of devices with diverse architectures. Furthermore, the transition of infrastructure from ground to space for internet service providers (ISP) necessitates the integration of non-traditional security enhancements. Assuming that all elements of the wireless systems are mobile, frequent authentication and key establishment are required for a large number of devices \cite{aman2020privacy}. It is evident that the current one-by-one methods may not provide a desirable solution for ensuring secure communication among devices in such systems.

Next-generation communication systems demand authentication methods that are scalable, lightweight, secure, and user-friendly. With numerous devices exchanging data simultaneously, hundreds may require authentication at once. Group authentication schemes, which enable the simultaneous authentication of multiple users while establishing a shared secret key with minimal communication and computational overhead, emerge as a practical solution to this challenge. Although several group authentication schemes have been proposed, a reliable implementation has yet to emerge. This is due to several reasons, some of which are outlined below:
\begin{itemize}
\item Group authentication schemes based on polynomial interpolation are susceptible to denial-of-service (DoS) attacks. Specifically, these schemes fail to detect intruders attempting to join the process. Unfortunately, an unauthorized entry results in the termination of the authentication process. 
\item The most recent promising group authentication scheme (GAS) \cite{Guzey}, which employs inner product spaces, lacks the capability to identify the participants involved in the process.
\item Each member of the group must have the ability to identify malicious actors. 
\end{itemize}

The proposed scheme defines processes for key generation, group key generation, group authentication, the inclusion of new members by any existing member, and the identification of malicious actors. For group key generation, users compute a shared key by performing projection operations using their private information combined with public data, which is broadcast by a manager or any participant prior to the authentication process. Interestingly, all participants arrive at the same result despite utilizing their unique secrets. This approach guarantees private and autonomous key establishment, eliminating the necessity for direct data exchange.

The group manager (GM) generates keys by defining a subspace $W$ within a universal space $E$ and constructing a random polynomial $f(x)$. Each group member is subsequently assigned a unique private key, which is derived from the polynomial and the selected basis of the subspace $W$. User authentication is accomplished through the member's secret, comprising a basis and specific details about the function. Additionally, the new member addition feature allows an existing member to securely add another user in the absence of the $GM$. While the scheme effectively prevents non-members from accessing the communication channel, it also incorporates a mechanism for detecting malicious members. This is achieved using a recursive subgroup division algorithm to isolate and identify adversarial members within the group.

The proposed method recommends performing group authentication using publicly shared information. This approach removes the necessity for members to exchange private information, thereby minimizing security risks and reducing communication overhead. Furthermore, it facilitates the identification of malicious actors and prevents unauthorized participation in the authentication or key establishment processes. Most importantly, this method enables group members to verify one another, ensuring that only authorized individuals can engage in the group’s communication. In summary, this proposed method offers the following contributions:

\begin{itemize}
    \item A method for group membership confirmation and group key establishment, where the process cost is independent of the group size.
    \item A method for authenticating members participating in the process.
    \item A method for detecting and identifying malicious members during the authentication process.
    \item A practical solution suitable for scalable environments.
\end{itemize}

The structure of this paper is organized as follows: Section II reviews related work and provides their evaluation. Section III covers the preliminaries. Section IV introduces the proposed scheme. Section V presents the security analysis and discussion. Section VI compares the performance of the proposed scheme with other group authentication methods. Finally, Section VII concludes with the contributions to the field and outlines future research directions.

\section*{Related Work}

The first-generation and the second-generation group authentication schemes are founded on polynomial interpolation over finite fields. This concept is inspired by secret-sharing schemes; therefore, we begin by presenting some fundamental information about secret-sharing schemes. The first practical methods for sharing a secret among some users are presented by Shamir \cite{Shamir} and Blakley \cite{Blakley}. Their approach ensures that when $t$ or more participants collaborate, they can reconstruct the secret using their respective shares. However, if fewer than $t$ participants are involved, no information about the secret can be obtained. In Shamir's proposal, the secret is divided into several pairs and distributed among the shareholders. Only those who possess a number of shares meeting or exceeding the threshold value can successfully retrieve the secret.  Blakley's proposal addresses the secret-sharing problem through the use of hyperplane geometry. Each of the $n$ participants is assigned a hyperplane equation within a $t$-dimensional space over a finite field. Each hyperplane occasionally intersects a specific point, and the intersection of these hyperplanes collectively represents the secret. To reconstruct the secret, participants must solve the corresponding system of equations. Unlike Shamir's secret-sharing scheme, Blakley's approach may not be practical for certain use cases due to efficiency constraints \cite{Menezes}.

Chaum and Van Heyst \cite{chaum1991group} explore four distinct techniques for group signatures. Their findings are highly significant in the field of group authentication, which is why we have chosen them to be presented in this work.

\begin{itemize}
\item  In the first group signature scheme, a trusted authority selects a public key system and assigns a unique private key to each group member. The authority then makes a publicly accessible list of corresponding public keys. When a group member signs a message using their private key, the recipient verifies the signature by checking it against the public key list.

\item The second group signature scheme, which is based on an RSA factorization problem, allows group members to sign anonymously and enables the verification of these signatures. The scheme preserves the anonymity of the signer but also provides an option to reveal their identity when necessary. In this approach, a trusted authority selects two large prime numbers, $p$ and $q$, and computes $N = p \times q$. The modulus $N$ is then made public. A function $f$ is chosen such that its output is coprime with $N$. Each group member receives a secret key $s_i$, which is a large random prime. The authority computes $v = \prod s_i$ and publishes $N, v,$ and $f$. To generate a signature, a member transforms the message $m$ using $f$ and computes $S = (f(m))^{s_i} \mod N$. A zero-knowledge proof is then provided to verify the signature's validity without revealing the identity of the signer.

\item The third group signature scheme is similar to the second one but introduces a "Trusted Public Directory" containing each member's RSA modulus, $N_i = p_i \times q_i$. Each member's secret key consists of their factors of their modulus, $p_i$ and $q_i$. During the setup phase, a trusted authority generates an independent RSA modulus $N = p \times q$. To sign a message, a member randomly selects a group of participants (including themselves) and constructs the signature using their private key’s prime factor $p_i$, $S = (f(m))^{p_i} \mod N$. The member then provides a zero-knowledge proof to show that $p_i$ is a divisor of the product of the selected members' moduli without revealing $p_i$. Verification is also carried out using zero-knowledge proofs to ensure anonymity.

\item The fourth group signature scheme is based on a large prime $p$ and modular arithmetic, enabling group members to sign messages anonymously. Each member has a secret key $s_i$ along with a corresponding public key $k_i = g^{s_i} \mod p$, where $g$ and $h$ are public generators. To sign a message $m$, a member selects a random subset of members. The signature is then computed using the member’s secret key as $S = m^{s_i} \mod p$. The signer subsequently provides a zero-knowledge proof to demonstrate that the signature corresponds to a valid public key while maintaining their anonymity.

\end{itemize}

Harn's work \cite{Harn} illustrates an efficient application of Shamir’s secret-sharing scheme within the context of group authentication. As opposed to conventional authentication schemes, which require one-to-one authentication between each node and incur high communication costs, this approach significantly reduces communication overhead. The study that focuses on groups is specifically intended for group authentication. Based on Shamir's $(t, n)$ secret sharing scheme (SSS), it proposes a basic $t$-secure $m$-user $n$-group authentication scheme $(t, m, n)$, where $t$ is the proposed scheme's threshold, $m$ is the number of users who participated in the group authentication, and $n$ is the total number of group members. Harn introduced three types of group authentication schemes, and various studies in the literature have analyzed the security of each scheme. For instance, two of the schemes are vulnerable to an attack similar to a replay attack. Specifically, if the scheme allows multiple authentication attempts, an attacker could repeatedly execute the process to extract system secrets and users' private tokens. Conversely, if each secret is restricted to a single authentication attempt, regardless of its success or failure, the system becomes susceptible to DoS attacks. In this scenario, a malicious entity could inject a false value, disrupting both the authentication process and the overall group authentication mechanism.

Harn's work has served as a foundation for subsequent studies, including \cite{Li} and \cite{Mahalle}. These studies build upon Harn’s approach, introducing alternative methods. For instance, in \cite{Li}, Li \textit{et al.} applied Harn’s scheme in conjunction with ECC-pairing to implement group authentication and key agreement within the LTE network. Later, Mahalle \textit{et al.} proposed a solution in \cite{Mahalle} leveraging Paillier threshold cryptography as a core tool. Both studies also include a comparative analysis of their performance relative to Harn’s work.

A physically unclonable function (PUF) is defined as a representation of a unique and unclonable characteristic inherent to a physical object. An ideal PUF functions as a one-way mechanism, where its output consistently depends on the physical system. It is easy to evaluate and construct, operates like a random function with unpredictable outputs, and is inherently unclonable \cite{gope2019lightweight}. PUFs are extensively utilized in various group authentication studies, such as \cite{Ren}. In \cite{Ren}, the protocol comprises two key steps: the registration phase and the mutual authentication and key agreement phase. To achieve mutual authentication and key agreement, the protocol employs a group authentication and data transmission technique for NB-IoT, using the PUF’s output as a shared root key. Additionally, PUFs have been applied in conjunction with the Chinese Remainder Theorem (CRT). Studies, including \cite{Singh}, have explored the integration of CRT with PUF. For instance, the study in \cite{Singh} introduces a lightweight key distribution and group authentication scheme that combines CRT, factorial trees, and PUF. An enhanced version of this work is further presented in \cite{Singh}.

Blockchain technology has seen widespread adoption across various fields. In \cite{Zhang}, if a new block receives a valid group aggregate signature from the group to which the block author belongs, it will be considered legitimate. Furthermore, the study \cite{zhang2019group} provides a detailed explanation of the authentication and key exchange processes when mobile devices join or leave blockchain-based mobile edge computing (BMEC) and a new block will be considered valid only if it obtains a valid aggregate group signature from the group to which the block creator belongs.

The works of \cite{Harn}, \cite{aydin}, and \cite{Guzey} have significant contributions to this field, with each being considered as defining distinct generations of group authentication. Consequently, their studies are detailed below.
\subsection*{First Generation Group Authentication Scheme}

One of Harn's group authentication methods is a synchronous $(t, m, n)$ scheme, where all participants must reveal their secret tokens at the same time. If they fail to do so, an unauthorized participant might forge valid tokens by leveraging the tokens already revealed by others. The other two methods include the asynchronous $(t, m, n)$ group authentication scheme and the asynchronous $(t, m, n)$ scheme with multiple authentication attempts. All of these schemes consist of two main phases: token generation and group authentication.

\subsubsection*{Basic (Synchronous) (t, m, n) group authentication:}

This fundamental scheme is intended for synchronous communication environments, where all members are required to engage and respond within a specified time window.

\subsubsection*{Asynchronous (t; m; n) group authentication:}

In this method, a $(t, m, n)$ group authentication scheme is proposed, enabling $m$ users (where $t \leq m \leq n$) to asynchronously release their values during a group authentication process.

\subsubsection*{Asynchronous (t; m; n) group authentications scheme for multiple authentication scheme:}
%\vspace*{0.3cm}

The $(t, m, n)$ asynchronous scheme with multiple authentications enables tokens to be reused for multiple authentications.

\subsection*{Second Generation Group Authentication Scheme}

In the study \cite{aydin}, the elliptic curve discrete logarithm problem (ECDLP) forms the basis of the proposed group authentication algorithm. The algorithm operates in two distinct stages: Initialization and Confirmation, each of which is described in detail below.

\textit{The Initialization Phase}

 \begin{itemize}
 \item  The GM selects a cyclic group $G$ along with a generator $P$ for $G$. Additionally, the GM determines the encryption algorithm $E = Encryption(\cdot)$, the decryption algorithm $D = Decryption(\cdot)$, and a hashing function $H(\cdot)$. A polynomial of degree $t - 1$ is also chosen by the GM, with the constant term defined as the group key $s$.

    \item The GM selects a public key $x_i$ for each user $U_i$ and generates the corresponding private key $f(x_i)$ for $i = 1, \dots, n$.
    \item The GM calculates $Q = sP$.
    \item The GM publishes $P$, $Q$, $E$, $D$, $H(s)$, $H(\cdot)$, and $x_i$, while ensuring that $f(x_i)$ is shared exclusively with the respective user $U_i$ for $i = 1, \dots, n$.
 
\end{itemize}

\textit{The Confirmation Phase}
 \begin{itemize}
    \item Each user computes $f(x_i)P$ and sends $f(x_i)P \parallel ID_i$ to the GM and other users, where $ID_i$ is the identification number of the user, and $\parallel$ represents the concatenation of two values.
    
    \item If the GM is not involved in the verification process, any user in the group computes:
    $$
    C_i = \left( \prod_{\substack{r=1 \\ r \neq i}}^{m} \frac{-x_r}{x_i - x_r} \right) \bigg(f(x_i)P\bigg)
    $$
    for each user, where $m$ represents the number of users in the group, and $m$ must be greater than or equal to $t$.
    
    \item The user verifies whether:
    $$
    \sum_{i=1}^{m} C_i \stackrel{?}{=} Q 
    $$
    holds. If this condition is satisfied, the authentication process is successfully completed. Otherwise, the process needs to be restarted from the initialization phase.
\end{itemize}

\subsection*{Third Generation Group Authentication Scheme}

In this method a novel  mathematical tool, inner product space, has been employed to facilitate group membership confirmation.
Every member has a unique basis for a specific subspace, $W$. This basis enables them to verify their group membership. The nature of subspace $W$ allows one to select infinitely many bases for it, but knowing any basis for $W$ is sufficient to obtain the group’s secret key. However, the algorithm’s design ensures that breaking the group’s authentication scheme requires knowing the chosen basis, which the group manager $GM$ keeps secret. 

The group manager $GM$ employs a randomly selected function $ f(x) $ during the distribution of secrets. This function $f(x)$ can be a polynomial of degree $ d $, which the security analysis suggests should be larger than the expected number of users in group $G$.

Any user $U_i$ in group $G$ is assigned a public key $x_i$ (preferably an integer) and a secret key 
\[
B_i = \{f(x_i)v_1, f(x_i)r_1v_2, f(x_i)r_2v_3, \ldots, f(x_i)r_{n-1}v_n\}.
\]
\noindent Here, $r_1, r_2, \ldots, r_{n-1}$ are random numbers selected by the group manager ($GM$), and these values remain the same for all users.

Each group member’s secret $B_i$ is linearly independent within $ W $ and serves as their secret key. This ensures that the private information of each user is independent of others.

The group manager or any member publishes two random vectors $v$ and $g$. Participants who successfully compute the inner product of $g$ with the projection of $v$, $\mathrm{Proj}_Wv$ are confirmed as group members. Below, we present the steps of the algorithm.
\begin{itemize}
   
    \item A random element $v \leftarrow E$ is selected by the group manager or a group member.
    \item A nonce vector $g \leftarrow E$ is also selected and published alongside $v$.
    
    \item Each user $U$ computes the inner product $\langle g, \mathrm{Proj}_W v \rangle$ to derive a shared secret.
    
    \item The shared secret $s$ is calculated as:
    $$
    s \leftarrow \langle g, \mathrm{Proj}_W v \rangle
    $$
    
    \item Finally, the user $U_i$ releases the requested bits of $s$ for further verification.
\end{itemize}

 In this work, we enhance group authentication by introducing additional flexibility through the use of inner-product spaces. To this end, we begin by presenting the mathematical framework underlying the scheme.

\section*{Preliminaries}

\subsection*{Lagrange Interpolation Polynomials}

Lagrange interpolation describes a method for constructing the unique  polynomial $p_n(x)$ of degree less than $n$  that satisfies the $p_n(x_i) = y_i$ for each $i = 1,..., n$  given a set of points $(x_i, y_i)$.  

The equation is:

$$p_n(x)=\sum_{j=0}^{n}{y_i}{\mathcal{L}}_{n,j}(x)$$

The cardinal functions ${\mathcal{L}}_{n,j}(x) $ satisfy:
$$ {\mathcal{L}}_{n,j}(x) = \prod_{k=0,k \neq j}^n \frac{x-x_k}{x_j-x_k}  $$
and
$${\mathcal{L}}_{n,j}(x_i) =\begin{cases} 1\quad if & i=j \\
                     0\quad if &  i\neq j 
       \end{cases}$$

In order to find $p_n(0)$, we use:
$$p_n(0)=\sum_{j=0}^{n}y_j\prod_{k=0,k \neq j}^n \frac{-x_k}{x_j-x_k}$$

\subsection*{Inner Product Space}
A vector space E over a field $\mathbb{F}$ with an inner product is called an inner product space. An inner product is denoted by
$${\displaystyle \langle \cdot ,\cdot \rangle :E\times E\to \mathbb{F}}$$

An inner product $\langle \cdot ,\cdot \rangle$ on a vector space E is an assignment that for any two vectors
$u, v \in E$ , there is a real number $\langle u,v \rangle$ satisfying the following properties:

\begin{enumerate}
    \item \textbf{Linearity}: $\langle au + bv,w \rangle = a\langle u ,v \rangle + b\langle v,w \rangle $.
    \item \textbf{Symmetric Property}: $\langle u,v \rangle = \langle v,u \rangle$
    \item \textbf{Positive Definite Property}: For any $u \in E$ , $\langle u,u \rangle \geq 0 $ and $\langle u,u \rangle = 0 $ if and only if $u=0$.
\end{enumerate}

The following observation presents the concept of utilizing inner-product spaces in group authentication schemes.
\begin{theorem}\label{indepthm}
Let $V$ be a vector space of dimension $n$ over a real numbers. The probability that the randomly selected $d\le n$ vectors is linearly dependent is negligible.   
\end{theorem}

The following theorem presents a method to find the unique projection of a vector onto any space whose basis is known. 
\begin{theorem}\label{thm:ortho}
 Let $S = \{b_1, . . . , b_n\}$ be an orthogonal basis for a vector space E . Then every vector $ w \in V$ can be written uniquely as a linear combination of vectors in the basis S. In fact, if  $$w=c_1b_1 + c_2b_2 + ... + c_nb_n  $$ then $$c_j = \langle w,b_j \rangle = \frac{w \cdot b_j}{b_j \cdot b_j} $$ \end{theorem}

\iffalse

\fi
If $S$ is not an orthogonal basis, the Gram-Schmidt method is used to find an orthogonal basis from $S$.

\subsection*{Gram-Schmidt method}
The Gram-Schmidt process is a method for converting a set of vectors into an orthogonal basis \cite{hoffmann1989iterative}. 

\begin{theorem}\label{thm:gram}
Given a basis $\{x_1, \ldots, x_p\}$ for a nonzero subspace $W$ of $\mathbb{R}^n$, define
\[
\begin{aligned}
v_1 &= x_1, \\
v_2 &= x_2 - \frac{x_2 \cdot v_1}{v_1 \cdot v_1} v_1, \\
v_3 &= x_3 - \frac{x_3 \cdot v_1}{v_1 \cdot v_1} v_1 - \frac{x_3 \cdot v_2}{v_2 \cdot v_2} v_2, \\
&\vdots \\
v_p &= x_p - \frac{x_p \cdot v_1}{v_1 \cdot v_1} v_1 - \frac{x_p \cdot v_2}{v_2 \cdot v_2} v_2 - \cdots - \frac{x_p \cdot v_{p-1}}{v_{p-1} \cdot v_{p-1}} v_{p-1}.
\end{aligned}
\]

Then $\{v_1, \ldots, v_p\}$ is an orthogonal basis for $W$. In addition,
\[
\text{Span}\{v_1, \ldots, v_k\} = \text{Span}\{x_1, \ldots, x_k\}.
\]

\end{theorem}
The theorem ensures that any nonzero subspace possesses an orthogonal basis. 

\section*{Authentication Scheme}
To ensure clarity, we provide two definitions that aim to eliminate any potential ambiguity in the study of group authentication schemes:
\begin{definition}\label{GrupAuth}
A group authentication scheme (GAS) is a method that enables the simultaneous verification of multiple users belonging to a specific group.
\end{definition}

\begin{definition}
A fully-functional group authentication scheme (FGAS) simultaneously confirms the identities of many users, in addition to their membership in a specific group.
\end{definition}
The proposed method primarily focuses on the concept of FGAS, exploring its fundamental principles and potential advantages in secure authentication systems. The scheme is designed based on the following system and network model.

\subsection*{System Model}
The players in the scheme:
\begin{itemize}
	\item $G$: The set of users.
	\item $U_i$: Any member in the group in G.
	\item $GM$: The group manager likely has superior computational capabilities compared to any other member.
\end{itemize}

\begin{figure}[h!]
    \centering
	\includegraphics[width=5.7cm, height=4.7cm]{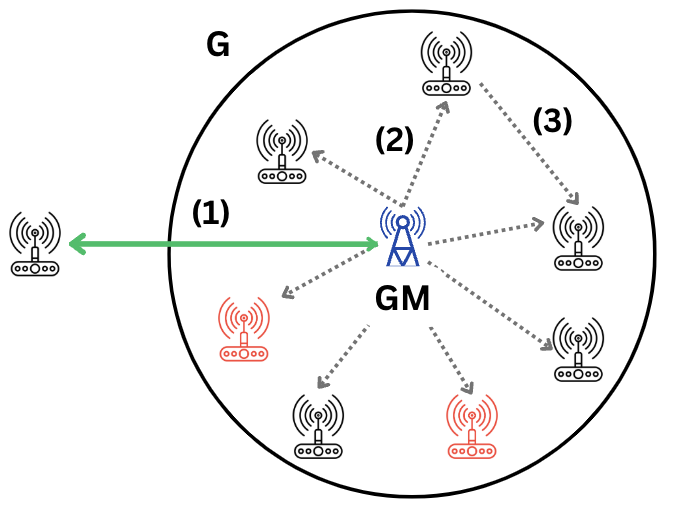}
	\caption{Communication Model: Channel (1) is designated as a dedicated line, while channels (2) and (3) are publicly accessible. }
	\label{fig:drone_communication}
\end{figure}

The registration of a user $U_i$ to the group $G$ is handled by the group manager $GM$. Therefore, we assume there is a secure channel between $U_i$ and $GM$. On the other hand, there might be more than one group manager or a member $U_j$ might act as a manager for certain situations. Apart from the registration phase, all other communication channels are assumed to be open to public as depicted in Figure \ref{fig:drone_communication}.

\subsection*{Network Model}
The basic structure of the network consisting of 3 channels:
\begin{enumerate}
    \item The first communication channel is between a user and the group manager responsible for handling registration. In this scenario, we are assuming it is the user's first time requesting to join the network.
    \item The second communication channel allows communication between the $GM$ and a user in the group so that the $GM$ can communicate directly with the group members.
    \item The third communication channel is designated for users to interact and communicate with one another. 
\end{enumerate}
 
Group members can communicate with each other over the publicly accessible channel. As previously stated, the only secure channel is the one used for the initial registration of a member. We assume that the group manager responsible for the first registration possesses superior computational and communication capabilities. The red color users are  non-legitimate users that need to be detected. Additionally, we consider the communication channel to be reliable, meaning that no bit errors occur during data exchange between any two entities.

\subsection*{Key Generation}
Typically, the responsibility for distributing the keys of a user inside a group is given to the group manager. For each group, the scheme utilizes subspaces from a pre-determined universal inner product space $E$. One possible choice for $E$ is an infinite-dimensional vector space; for instance, $E$ could consist of all polynomials over a finite field  $\mathbb{F}$ .\\
%KeyGen:\\
$KeySpace$= Any subspace $W$ of $E$.\\ 
$KeySpace_n$= Any subspace $W$ of $E$ with a dimension $n$. 

The dimension $n$ should be small in situations with limited memory. However, the value of $n$ is also tied to the security parameter of the scheme, so it must be chosen cautiously, considering the trade-off between security and cost factors. We will investigate the selection process of $n$ in the coming sections. 

The first task of the manager is to select a suitable subspace $W$ for the group $G$.  For the sake of simplicity, we select the universal space $E$ as $\mathbb R^d$ for a natural number $d$. In this case, the group manager choice for dimension, $n$, of $W$ should be less than $d$. In other words, $$KeySpace_n=\text{ Any subspace } W \text{ of } \mathbb R^d \text{ with dimendion } n< d.$$

The $GM$ determines the subspace $W$ by randomly selecting $n$ vectors in $E$. In other words, $GM$ selects and keeps the set $B = \{v_1,...,v_n\}$ for the subspace $W$ as a secret. Theorem \ref{indepthm} implies that the set $B$ is linearly independent and as $n<d$, we can definitely conclude that the set $B$ is basis for $W$.  The group manager $GM$ also employs  randomly selected function $f(x)$ is of degree $1$ while distributing the secrets of members.  Any user $U_j$ in the group $G$ is given a public key $x_{j}$ which in general is selected to be an integer. The user's, $U_j$, private key is obtained via its public information $x_j$, the group manager's basis and the functions that were selected by the group manager:
 $$B_{j} = \bigg \{ f(x_{j} )v_{1}, f(x_{j})v_{2}, f(x_{j})v_{3}, \ldots, f(x_{j} )v_{n} \bigg \}$$

The private information of each user is unrelated to that of others i.e., one user can not construct any of other members' private set. Algorithm 1 provides a step-by-step outline of how the $GM$ generates a unique key for every user.

\begin{algorithm}
\caption{KeyGen: Key Generation}\label{alg:cap}
\begin{algorithmic}
\\
\State \textbf{Require:}
\State $B : \{v_{1}, v_{2}, v_{3}, . . . , v_{n}\} \gets$ random $n$ vectors in $E$.
\State $f(x) \gets $ random element in $\mathbb{F}[x]$ of degree $1$.

\State \textbf{Key Generation}:
\State Public key: $U_j \gets x_j :x \in \mathbb{F}$.
\State Private key: $U_j \gets B_j : \{f(x_j )v_1, f(x_j)v_2, . . . , f(x_j)v_n\}$.
\end{algorithmic}
\end{algorithm}

\subsection*{Group Key Generation}

We present an algorithm designed to generate  a secret key for the group members. For a given subspace $W$, each member has a distinct basis.  A member chooses two vectors, $v$ and $h$, randomly from the vector space $E$, with a preference for $v$ and $h$ not belonging to $W$. If  $v$  is already part of  $W$, the projection operation would return the vector itself, which would not be meaningful from a computational perspective. The vectors $v$ and $h$ are publicly disclosed, and the $s$ is derived from them by calculating $$s = \langle Proj_W v, h \rangle$$

Computation of  $s$ requires knowledge of a basis for the subspace $W$ and the projection of $v$ onto the subspace of $W$ remains same regardless of  a basis for $W$. The projection of $v$ onto $W$ is given by  $\mathrm{Proj}_Wv$ , which can be computed as the sum of the projections of $v$ onto each basis vector of  $W$ .

\begin{algorithm}
	\caption{GKeyGen: Group Key Generation}\label{alg:cap}
	\begin{algorithmic}
    \\
            \State \textbf{Require:} Two random vectors $v, h \in E$.
		\State \textbf{Group Key Generation}:
		\State Each member $U_j$ computes $s \gets \langle \mathrm{Proj}_{B_j}v, h \rangle$.

	\end{algorithmic}
\end{algorithm}

Each user can compute $s$ by executing a single projection operation and utilize it as a key. This key can then be employed for secure communication within the group.

\subsection*{Group Authentication}
The purpose of group authentication is to verify the membership of many users in a certain group. The $GM$ establishes the set $B = \{v_1, \ldots, v_n\}$ and the functions $ f(x) $. The public key $x_j$ of a user $U_j$ is known by everyone. Each user knows their own set $B_j$ and the value $x_j$:

\begin{center}
\begin{tabular}{ |c||c| } 
 \hline
 Public key &   Basis for each user  \\ 
 \hline
 $x_1$  & $B_1=\{f(x_{1})v_{1}, f(x_{1})v_2,\dots , f(x_{1} )v_n\}$  \\ 
 \hline
 $x_2$  & $B_2=\{f(x_{2} )v_{1}, f(x_{2})v_2,\dots , f(x_{2} )v_n\}$  \\ 
 \hline
 ...   & ...  \\ 
 \hline
 $x_j$  & $B_j=\{f(x_j )v_{n}, f(x_j)v_2, \dots, f(x_j )v_n\}$  \\ 
 \hline
\end{tabular}
\end{center}

Lagrange Interpolation gives the following equation:

\begin{align}
f(0)=\sum_{j=0}^{r}{f_i(x_j)}\prod_{k=0,k \neq j}^r \frac{-x_k}{x_j-x_k} \label{lag_1}
\end{align}

The operations that each user joining the process needs to complete are listed below.\\
\textbf{Part 1:} The user $U_j$ computes  $A_j$:

\begin{equation}
	A_j = \prod_{k=0, k \neq j}^r \frac{-x_k}{x_j - x_k} \label{lag_2}
\end{equation}
where $x_k$ represents the public identities of the members joining the process.\\
\textbf{Part 2:} Each user multiplies the result of its computation,  $A_j$, by a basis vector whose position in the set is either predetermined or agreed upon during the process. If no prior agreement is made, they all use the first vector from their basis set. In other words, they all compute:
\[U_1: f(x_1) v_i A_1\]
\[U_2: f(x_2) v_i A_2\]
\[\vdots\]
\[U_j: f(x_j) v_i A_j\]
where we assume that they all know to use the $i^{th}$ position vector of their basis set.\\
\textbf{Part 3:}  Each user computes the inner product of their result with the random vector $g$  published by the group manager.

\[U_1: \langle f(x_1) v_i A_1, g \rangle\]
\[U_2: \langle f(x_2) v_i A_2, g \rangle\]
\[\vdots\]
\[U_j: \langle f(x_j) v_i A_j, g \rangle\]

Note that each computation gives $$\langle f(x_j)v_iA_j,g\rangle=f(x_j)A_j\langle v_i,g\rangle$$

Each user encrypts their result with the group key  $s$  and sends it to $GM$: 
$$Enc_{s}[f(x_j)A_j\langle v_i,g\rangle]$$

\textbf{Part 4:} The group manager, ($GM$), performs the following operation using the inputs received from the participants:
$$\sum_{i=1}^{r}f(x_j)A_i\langle v_i,g \rangle=\langle v_i,g \rangle \sum_{j=1}^{r}f(x_j)A_j$$

$GM$ confirms authentication if its computation gives the following:

$$\langle v_i,g \rangle \sum_{j=1}^{r}f(x_j)A_j  \stackrel{?}{=}  \langle v_i,g\rangle f(0)  $$

The above equality holds since by equations \eqref{lag_1} and \eqref{lag_2}, we have $$\sum_{j=1}^{n}f(x_i)A_j=f(0)$$

The following algorithm summarizes the steps each participant must perform for authentication.

\begin{algorithm}
	\caption{GroupAuth: Group Authentication}\label{alg:cap}
	\begin{algorithmic}\\
            \State \textbf{Require:} $x_j$, $B_j$, $g$
		\State Public key: $U_j \gets x_j : x \in \mathbb{F}$
		\State Private key: $U_j \gets \{f(x_j)v_1, f(x_j)v_2, \ldots, f(x_j)v_n\}$
		
		\State $g \gets$ random element in $E$
		\State \textbf{Group Authentication:}
		\State Each user computes 
		
		\State Part 1: $A_j$: $\prod_{k=0,k \neq j}^n  \frac{-x_k}{x_j-x_k}$
		\State Part 2: $U_j: f(x_j) v_i A_j$
		\State Part 3: $U_j: \langle f(x_j) v_i A_j, g \rangle$
		
		\State $GM$ receives results from each user and confirms authentication if the equation below holds.
		
		\State Part 4: $\langle v_i,g \rangle \sum_{j=1}^{n}f(x_j)A_j  \stackrel{?}{=}  \langle v_i,g\rangle f(0)$
		
	\end{algorithmic}
\end{algorithm}

\subsection*{New Member Added to the Group by a Member}
When the $GM$ is unavailable to handle the addition of new members to group, an existing member should have the capability to add a new member. This ensures that the new member can securely communicate with the rest of the group. Also, the $GM$ should be able to easily identify which member added it to the group. 

Note that $U_j$ has,
 $$B_{j} = \bigg \{ f(x_{j} )v_{1}, f(x_{j})v_{2}, f(x_{j})v_{3}, \ldots, f(x_{j} )v_{n} \bigg \}$$
$U_j$ selects a random number  $t$  and creates a new basis for $W$,
 $$B_{new} = \bigg \{ tf(x_{j} )v_{1}, tf(x_{j})v_{2}, tf(x_{j})v_{3}, \ldots, tf(x_{j} )v_{n} \bigg \}$$
 
The host $U_j$ does not need to know the function  $f(x)$  to include the new member in the group conversation. It is important to note that the new member can easily attend the group communication by utilizing its own basis.
 
\begin{algorithm}
	\caption{GroupMemAdd: New Member Added to the Group by a Member $U_j$}\label{alg:groupadd}
	\begin{algorithmic}\\
		\State \textbf{Require:} $x_j$, $B_j = \{f(x_j)v_1, f(x_j)v_2, f(x_j)v_3, \ldots, f(x_j)v_n\}$
		\State $t \gets$ Random number selected by $U_j$
		\State \textbf{Member Additon:}
		
		\State $U_j$ creates a new basis for $W$,
		\State $B_{new} = \{t  f(x_j)v_1, t  f(x_j)v_2, t  f(x_j)v_3, \ldots, t  f(x_j)v_n\}$

	\end{algorithmic}
\end{algorithm}

\subsection*{Malicious Actor Detection in Groups }

This algorithm is designed to identify a malicious actor within a group. It starts by dividing the group of  $n$  people into two subgroups, $G_1$ and $G_2$.  If either of them contains  malicious users, it is further divided into smaller subgroups, and the authentication is repeated. This process continues recursively until the subgroup size is reduced to two. If both subgroups are found to be free of them, they are combined into a single group. The algorithm ensures that all malicious actors are isolated and identified through systematic division and verification of subgroups. We should note here that unlike the first and second generation group authentication schemes, a non-member can not join the authentication process since it must have a basis for $W$ to create the group secret key $s$.

\begin{algorithm}
\caption{MalActDetect: Malicious Actor Detection in Groups}\label{alg:cap}
\begin{algorithmic}\\
\State \textbf{Require:} Number of people in the group, $n$, Number of enemies in the group (unknown)
\State \textbf{Initial:}

\State $n$ is the number of people in the group
\State Number of enemies in the group (unknown)
\State Goal: Find out who the enemies are

\State \textbf{Algorithm:}
\While{$n > 1$}
\State Divide $n$ into two subgroups: $G_1$ and $G_2$
\State Check whether each subgroup contains malicious actors by checking both subgroups
\If {$G_1$ contains malicious actors}
\State Divide $G_1$ into two subgroups and repeat the process
\EndIf
\If {$G_2$ contains malicious actors}
\State Divide $G_2$ into two subgroups and repeat the process
\EndIf
\If {both $G_1$ and $G_2$ are determined to be valid users}
\State Combine these subgroups into a single group
\EndIf
\EndWhile
\end{algorithmic}
\end{algorithm}

\section*{Security Analysis}

In this section, we examine cryptoanalysis and prominent threat models and assess the robustness of the proposed algorithm against these attacks.

\subsection*{Cryptanalysis}

The scheme has the following set up: \textit{KeyGen} algorithm determines a subspace  $W$ within a universal space $V$. Deciding a space $W$ means determining a suitable basis for it, denoted as $B=\{w_1,\dots,w_n\}$. Following this, the scheme dictates the selection of a function $f(x)$.  While the function can theoretically be of any degree, for mutual authentication to be feasible, the scheme enforces $f(x)$ to be linear. In other words, the key consists of two components: the basis of the chosen subspace $W$ and a function $f(x)=ax+b$  where $a$ and $b$ are randomly selected from the base field. This discussion leads us to the following description. 

\begin{definition}
The key space of the scheme is made up of all subspaces of the universal space $V$.
\end{definition}

\begin{theorem}\label{Thm_Inf}
Let $W$ be the subspace selected by the central authority. Define $R$ as a random variable representing the dimension of the space chosen by any central authority. For a randomly chosen $a$, we have
$$Pr(R=a)=\epsilon$$
where $\epsilon$ is a negligible quantity.
\end{theorem} 
\begin{proof}
The universal space $V$ has been selected to have infinite dimension. Therefore, every natural number has an equal probability of being the dimension of the selected subspace.
\end{proof}
In scenarios where the universal space selected for the scheme is infinite-dimensional, an external eavesdropper lacks the capacity to accurately infer the dimension of the chosen subspace \( W \). Theorem \ref{Thm_Inf} is introduced to formalize the constraints and considerations involved in selecting an appropriate universal space for the scheme.

As previously noted, a polynomial space over an arbitrary field \( \mathbb{F} \) constitutes a viable candidate for practical implementation. However, for the sake of simplicity and computational tractability, in our experiments we adopt the finite-dimensional Euclidean space \( V = \mathbb{R}^n \) as the universal space over the real field \( \mathbb{R} \). In this setting, the dimensionality of \( V \) is publicly known and thus accessible to an eavesdropper, specifically \( \dim(V) = n \). Moreover, the adversary may possess a non-negligible probability of correctly guessing the dimension of the selected subspace \( W \subset V \).

Despite this, the subsequent analysis (Theorem 5 and 6) demonstrates that such dimensional information does not compromise the secrecy of \( W \). Crucially, knowledge of the dimension alone does not yield any substantive insight into the structure or identity of the secret subspace. Even within the finite-dimensional case \( V = \mathbb{R}^n \), for any \( d < n \), there exist uncountably many distinct subspaces of dimension \( d \). This inherent abundance ensures that the dimension of \( W \), while potentially guessable, does not facilitate its identification or reconstruction by an adversary.

It is important to emphasize that, in practical applications, the scheme typically operates over low-dimensional subspaces. Consequently, the preceding theorem may have limited relevance in real-world scenarios, where the dimensionality of the selected subspace often falls below 100. Nonetheless, the following theorems establishe that an eavesdropper remains incapable of verifying whether a randomly selected dimension coincides with that of the actual secret subspace.

\begin{theorem}
Let \( V \) be the publicly known key space over a base field \( \mathbb{F} \). Then, an eavesdropper has no means of verifying whether a randomly guessed dimension coincides with the true dimension of the secret subspace \( W \subset V \).
\end{theorem}

\begin{proof}
The only information available to the eavesdropper consists of a finite collection of vectors in \( V \), which are indistinguishable from uniformly random elements in the absence of structural knowledge about \( W \). Since no additional data is revealed regarding the construction or dimension of \( W \), the adversary lacks any statistical or algebraic basis to confirm the correctness of a guessed dimension. Thus, any such guess remains unverifiable.
\end{proof}

\begin{corollary}
From an outsider's point of view, the scheme ensures an information-theoretically secure approach to authentication and key establishment. 
\end{corollary}

\begin{theorem}
    Assume that the eavesdropper is aware of the dimension of the subspace $W$. This knowledge does not leak any secrecy of the group communication to the eavesdropper. 
\end{theorem}

 \begin{proof}
 Let $d$ be the dimension of the subspace $W$. The universal space $V$ has infinitely many subspaces of dimension $d$. Therefore, knowing only $d$ does not reveal any information about the specific subspace $W$ that has been selected.    
 \end{proof}

The previous discussion holds from an outsider’s perspective. Next, we provide an analysis of the scheme from an insider's point of view. The following information is available to group members.

\begin{enumerate}
    \item A basis $B$ for the subspace $W$.
    \item Two publicly known vectors, $v$ and $h$, are broadcast during the authentication and key generation phases. 
\end{enumerate}

\begin{theorem}
 A legitimate user $U$ possesses its own basis $B_u$, and it is infeasible to derive group secret function $f(x)$ or the group secret basis  $B$.  
\end{theorem}

\begin{proof}
Each vector in the basis of the user is a multiple of $f(x_u)$ where $x_u$ is the public identity of the user $U$. As $u$ does not have the knowledge of the function $f(x)$, retrieving the vectors in the secret group basis is impossible.   
\end{proof}

One might argue about what happens when two or more users get together to obtain the function in order to access the group manager's secret. We discuss this and other situations below.

\subsection*{Sybil Attack}

Douceur \cite{douceur2002sybil} first introduced the concept of a Sybil attack in peer-to-peer networks. This attack involves an adversary generating multiple fake identities, known as Sybil entities, to impersonate numerous users, either concurrently or at different times. In the proposed work, the Sybil attack will be non-functional, given the proposition below.

\begin{proposition}\label{thm:prop1}
	Even if an adversary  $\mathcal{A}$ acquires the bases of multiple users, it still cannot determine the group manager's secret.
\end{proposition}

\begin{proof}
In practice, the inner product spaces are generally considered over large finite or real fields. Assume that the adversary $A$ has the following bases:  
$$
\{f(x_i)v_1, \ldots, f(x_i)v_n\}, \{f(x_j)v_1, \ldots, f(x_j)v_n\}
$$
Based on this knowledge, isolating \( f(x_i) \) from the basis elements is not possible. To determine the group manager's selected basis and the function \( f(x) \), it is necessary to distinguish \( f(x_i) \) and \( v_k \) from their product, \( f(x_i)v_k \). However, since both entities belong to an infinite field, they can take on any value. In fact, for each \( a \) in the base field, there is nothing preventing the assignment \( f(x_i) = a \), as there is no additional information to verify whether this assignment is correct.

\end{proof}

\subsection*{Denial of Service (DoS) Attack}
An earlier version of the group authentication method based on polynomial interpolation has a lack of identifying malicious actors. If one or more actors send incorrect results during authentication or key establishment, the process must be halted.  More critically, pinpointing the culprit during such an intervention is impossible. In other words, in case a presence of an attacker, neither authentication nor key establishment can be performed via the first and second generation group authentication schemes. Since detecting a malicious actor is not feasible, the entire system may eventually need to be shut down.

The third-generation group authentication scheme does not depend on the malicious actors for authentication and key establishment. As mentioned earlier, although the third-generation scheme prevents non-members from joining the process from the beginning, it has no ability to determine which users to involve in the process. In other words, it only verifies the membership of an entity. Our proposed work enhances the scheme by incorporating an additional feature that enables user authentication. That is, the proposed scheme not only confirms the membership of entities joining the process but also verifies their identities simultaneously. The structure of the scheme inherently prevents non-members from participating in the authentication process. Furthermore, it enables mutual authentication, allowing any member to easily detect a malicious actor when mutual authentication is applied. In the proposed scheme, only users possessing a valid basis set for a specific subdomain are allowed to participate in the key negotiation phase, and only valid users can perform authentication. This design effectively prevents attacker interference and safeguards the scheme against DoS attacks.

\subsection*{Existential Forgery Attack}

In this class of adversarial strategy, the attacker attempts to construct a valid basis that would enable unauthorized traversal of the authentication protocol. However, even when the secret space is constrained to a one-dimensional subspace, the attacker faces a fundamental obstacle: the universal space  $E$ admits uncountably many such subspaces, rendering the identification of the correct basis computationally and theoretically infeasible.

Moreover, possession of the group secret does not confer any meaningful advantage to the attacker. This is because the group secret alone does not reveal structural information about the underlying space as any subspace of $E$  regardless of its construction, retains the potential to yield the final secret through legitimate protocol execution \cite{Guzey}. Thus, the entropy and ambiguity inherent in the space selection process act as a robust defense against basis reconstruction.

Even in the hypothetical scenario where an attacker manages to obtain a legitimate basis, this basis remains unusable for authentication purposes. The reason lies in the individualized integration of the key pair $(x_i, f(x_i)) $ into each participant's basis by the $GM$. During the authentication phase, the $GM$ performs a verification step by checking whether the public key $x_i$ corresponds to its embedded private counterpart  $f(x_i)$. Since the function $f $ is known exclusively to the $GM$ and is not derivable from public or private information, the attacker is unable to replicate or validate the required key pair. Consequently, the authentication process is effectively safeguarded against impersonation or unauthorized access.

\subsection*{Formal Security Model}

\begin{definition}[Game-Based Security Model]

The traditional game-based approach to cryptographic security defines an experiment between an adversary and a challenger to evaluate the security guarantees of a protocol. In this framework, the adversary $\mathcal{A}$ interacts with the system by issuing queries intended to reveal sensitive information, such as user secrets or session keys. A logical predicate $P$ specifies the \emph{bad event} the condition under which the adversary is considered successful. The probability of this event occurring quantifies the adversary’s advantage.

The formal game model introduced in~\cite{brzuska2011composability} defines a security game $G$ that operates over a structured state:
\[
\text{State} = (\text{LSID}, \text{SST}, \text{EST}, \text{LST}, \text{MST})
\]
where:
\begin{itemize}
    \item $\text{LSID}$: Local session identifiers,
    \item $\text{SST}$: Session-specific state,
    \item $\text{EST}$: Global protocol-related information,
    \item $\text{LST}$: Local session state,
    \item $\text{MST}$: Global game-related metadata.
\end{itemize}

The game is further characterized by the tuple:
\[
(\textit{setupE}, \textit{setupG}, Q, \textit{Valid}, \chi, P)
\]
where:
\begin{itemize}
    \item $\textit{setupE}$ and $\textit{setupG}$ initialize the environment and global state respectively,
    \item $Q$ is the set of allowed adversarial queries (e.g., \texttt{Send}, \texttt{Reveal}, \texttt{Corrupt}, \texttt{Test}),
    \item $\chi$ is the behavior algorithm that processes queries,
    \item $\textit{Valid}$ is a predicate that determines whether a query is admissible,
    \item $P$ is the bad event predicate.
\end{itemize}

Each query from $\mathcal{A}$ is processed by $\chi$, and a response is returned only if the query satisfies the $\textit{Valid}$ predicate.

The security experiment is denoted by:
\[
\text{Exp}^{G}_{\mathcal{A}, \pi}(1^\lambda)
\]
where $\pi$ is the protocol under evaluation and $\lambda$ is the security parameter. The output of the experiment is a bit $b \in \{0,1\}$, indicating whether the adversary has succeeded in triggering the bad event $P$.

\end{definition}

We extend the traditional game-based cryptographic model to capture the notion of \emph{basis vector leakage}, which is critical in protocols where vector-based secrets underpin security. If an adversary $\mathcal{A}$ obtains any  of a user's basis set, the confidentiality of the protocol may be compromised. To formally analyze this threat, we define the \emph{basis-secrecy game} $G^{\mathsf{BSec}}$, which models the adversary’s ability to distinguish genuine basis vectors from random noise.

\vspace{1em}
\textbf{State.} The game operates over the structured state:
\[
\text{State} = (\text{LSID}, \text{SST}, \text{EST}, \text{LST}, \text{MST})
\]

\vspace{1em}
\textbf{Game Definition.} The basis-secrecy game is defined by the tuple:
\[
G^{\mathsf{BSec}} = (\textit{setupE}, \textit{setupG}, Q, \textit{Valid}, \chi, P_{\mathsf{BSec}})
\]

\vspace{1em}
\textbf{Setup Phase.} The challenger runs $(\textit{setupE}, \textit{setupG})$ to generate the game environment and common parameters. The internal state is initialized as $(\text{LSID}, \text{SST}, \text{EST}, \text{LST}, \text{MST})$.

\vspace{1em}
\textbf{Query Phase.} The adversary $\mathcal{A}$ interacts with the challenger by issuing queries from the set $Q$. Each query is processed by the algorithm $\chi$, which returns a response only if the query satisfies the $\textit{Valid}$ predicate.

\vspace{1em}
\textbf{Challenge Phase.} At some point, the adversary issues a challenge. To evaluate its success, the challenger samples a random bit $b \leftarrow \{0,1\}$ and responds as follows:
\begin{itemize}
    \item If $b = 0$, the challenger reveals a genuine basis set $B_j = \{w_1, \dots, w_n\}$, split into a known subset $\{w_1, \dots, w_r\}$ and the remaining vectors $\{w_{r+1}, \dots, w_n\}$.
    \item If $b = 1$, the challenger reveals the same known subset $\{w_1, \dots, w_r\}$ along with $n - r$ uniformly random vectors from the underlying vector space.
\end{itemize}

The adversary outputs a guess $b'$, and is said to win the game if $b' = b$. The adversary’s advantage in this game quantifies its ability to distinguish genuine basis vectors from random noise, thereby measuring the leakage resilience of the protocol.

\textbf{Leakage Bound.} Let $W_i$ be a vector space with basis $B_j = \{w_1, \dots, w_n\}$. For any adversary $\mathcal{A}$, the probability that $\mathcal{A}$ can correctly infer the remaining basis vectors $\{w_{r+1}, \dots, w_n\}$ given partial knowledge of $\{w_1, \dots, w_r\}$ is negligible. Formally,
\[
P\left( \bigwedge_{j=r+1}^n \left( w'_j \in W_i \,\wedge\, w'_j \not\parallel \{w_1, \dots, w_r\} \right) \right) = \mathsf{negl}(\eta)
\]
where each $w'_j$ is linearly independent of the known subset, and $\mathsf{negl}(\eta)$ denotes a negligible function in the security parameter $\eta$.

\begin{proof}
Let $W \subset E$ be a subspace of a real vector space $E$, where $\dim W < \dim E$. Since $E$ is defined over the field of real numbers $\mathbb{R}$, it admits uncountably many subspaces of any fixed dimension less than $\dim E$, including those of dimension $\dim W$.\\
Suppose an adversary is given $\dim W - 1$ linearly independent vectors $\{w_1, \dots, w_{\dim W - 1}\} \subset W$. Despite this partial knowledge, there still exist uncountably many distinct subspaces of dimension $\dim W$ that contain these vectors. Consequently, the probability of correctly identifying the exact subspace $W$ from among all such candidates is zero in the measure-theoretic sense.
This reflects the inherent ambiguity in reconstructing a full basis from partial information in high-dimensional real vector spaces.
\end{proof}

\subsection*{Formal Analysis Using the Scyther Tool}
The Scyther tool provides formal security analysis and is used to test the proposed group authentication method. 
With the language called SPDL of this program, the roles in the protocol can be embedded to this test environment and resistance to attacks can be observed through this tool. The code of proposed method is provided (see the Appendix), and as seen in the code, an environment with three users and one group manager has been designed. Thanks to this tool, it is proven that the proposed method passes the aliveness, weak agreement, non-injective agreement, non-injective synchronisation, and secrecy of information tests. According to the Figure \ref{fig:scyther_result}, the tool shows that the proposed method does not complete with another unintended entity during the process; it is weak, but there is still consensus between the parties. In addition, the protocol does not contain illegitimate values and preserves the secrecy of information.

The Scyther code is designed with three users and one $GM$. Although this setup represents a small-scale prototype, it provides valuable insight into how the system would behave in a large-scale environment. Our proposed method consists of four main parts: the first three correspond to the user side, while the fourth belongs to the $GM$. On the user side, each user independently computes values such as $A_j$, $f(x_j)v_i A_j$, and $\langle f(x_j)v_i A_j, g \rangle$ without interacting with other users; their communication occurs solely with the $GM$. In the final stage, the $GM$ verifies whether the users are legitimate group members. During authentication, all users follow the same procedure, meaning that $U_1$, $U_2$, and $U_3$ act as isomorphic copies of the same user role. As the number of users increases, they continue to behave in the same way as in this Scyther scenario. As the users do not communicate directly and the outputs of this phase are encrypted using group keys, enlarging the group size does not introduce any additional structural vulnerabilities to the attack surface. In summary, although an attacker may collect information from large number of users, the Scyther analysis remains unaffected.

\begin{figure}[h!]
	\centering
	\includegraphics[width=8cm, height=12cm,keepaspectratio]{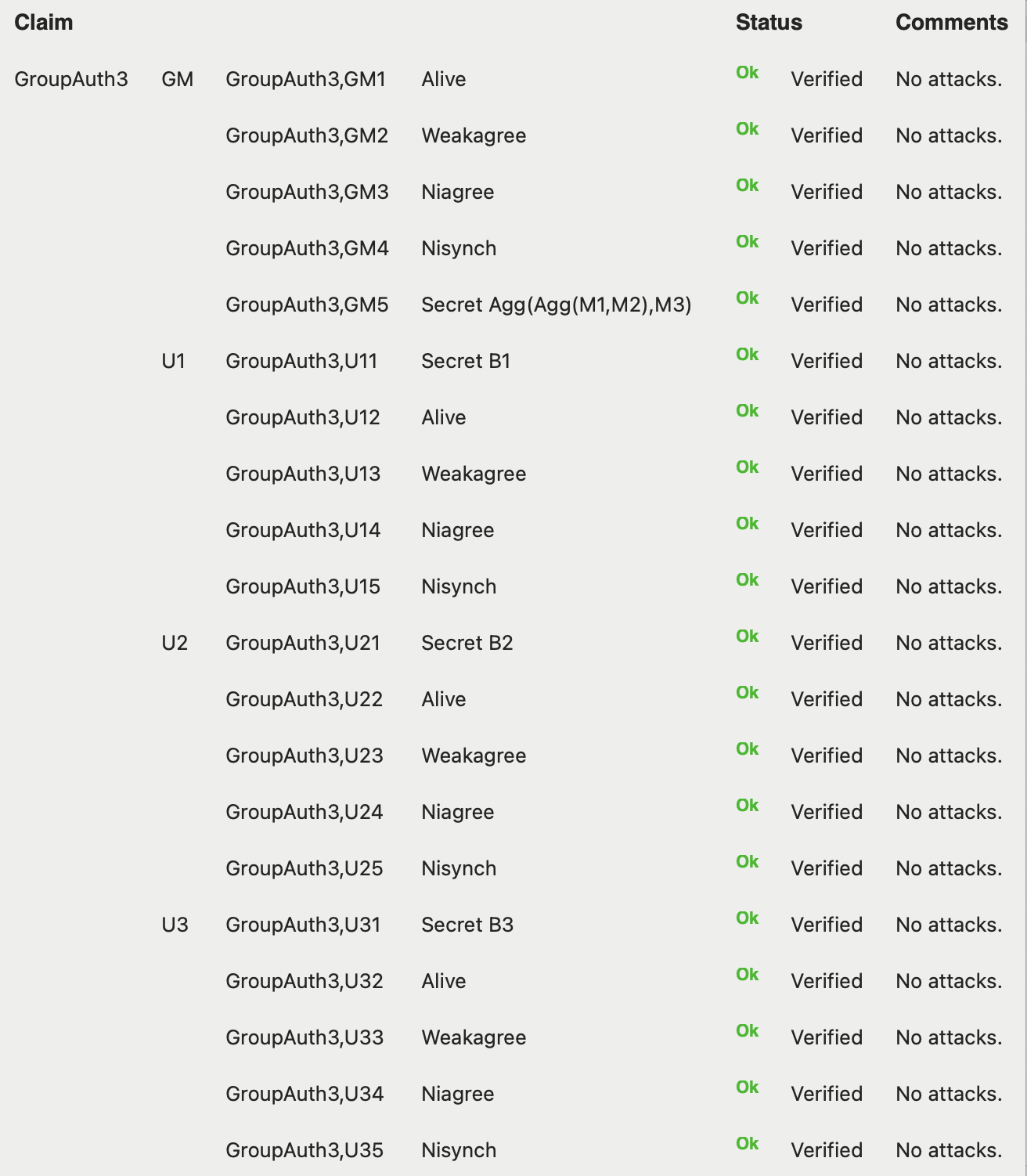}
	\caption{The Scyther tool results show that our proposed method is resistant to all attack scenarios considered by the tool.}
	\label{fig:scyther_result}
\end{figure}

\section*{Performance Analysis}

Real-time tests were conducted to evaluate the practicality of the proposed group authentication scheme using SageMath. The test series was performed on a computer running macOS, powered by an Apple M2 processor with 16 GB of RAM.  The comparison results in \cite{aydin} indicate that their approach requires less time and energy compared to others. While implementing their study, which utilizes elliptic curves, in our test environment, we adhere to their suggested implementation method. For \cite{semal2018certificateless}, we select the BLS12-381 elliptic curve because it is a well-known pairings-friendly elliptic curve. Therefore, the test environment uses the same elliptic curve for all other EC-based works \cite{aydin}, \cite{zhang2019pa}. In addition, the $py\_ecc$ Python library is employed, which enables the use of pairing operations.

This table \ref{tab:comp_cost} demonstrates the computational overhead of the benchmark methods, where $n$ refers to the number of users participating in the group authentication process.

\begin{table}[t!]
\centering

\begin{tabular}{|l|l|ccccccc|}
\hline  
Reference & Entity & Mult & Div & EC\_mult & EC\_add & B\_pairing & Mod\_exp & In\_prod \\
\hline
Semal et al. \cite{semal2018certificateless} & User & - & - & $2(n-1)$ & $n-1$ & - & $1$ & - \\
             & GM   & - & - & -        & -      & $n-1$ & $2(n-1)$ & - \\
            & Total   & - & - & $2(n-1)$   & $n-1$   & $n-1$ & $2n-1$ & - \\
\hline
Zhang et al. \cite{zhang2019pa} & User & $2(n-1)$ & - & - & - & - & - & - \\
             & GM   & $3n$ & - & $n+2$ & $1$ & - & - & - \\
             & Total   & $5n-2$ & - & $n+2$   & $1$   & - & - & - \\
\hline
Aydin et al. \cite{aydin}  & User & $n$ & $n-1$ & $1$ & - & - & - & - \\
             & GM   & - & - & $1$ & - & - & - & - \\
             & Total   & $n$ & $n-1$ & $2$   & -   & - & - & - \\
\hline
Our work     & User & $n$ & $n-1$ & - & - & - & - & $1$ \\
             & GM   & $1$ & - & - & - & - & - & $1$ \\
             & Total   & $n+1$ & $n-1$ & -   & -  & - & - & $2$  \\
\hline
\end{tabular}
\caption{Comparison of computational costs for user and $GM$ in different protocols. 
In this table, \textbf{Mult} denotes multiplication, \textbf{Div} denotes division, 
\textbf{EC\_mult} is elliptic curve point multiplication, \textbf{EC\_add} is elliptic curve point addition, 
\textbf{B\_pairing} denotes bilinear pairing, \textbf{Mod\_exp} is modular exponentiation, and \textbf{In\_prod} represents inner product computation.}
\label{tab:comp_cost}
\end{table}

These operations can be sorted by complexity among them by computing the each one single operation. This table \ref{tab:avg-times} is arranged from the most costly to the least costly. We did not analyze field multiplication/division and real number multiplication/division separately, as their differences are negligible. In addition, the execution time of the inner product operation is calculated based on a 10-element vector.

\begin{table}[t!]
\centering

\begin{tabular}{l c}
\hline
Operation & Average Time (ms) \\
\hline
Bilinear pairing    & 352.361 \\
Modular exponentiation & 18.711 \\
EC multiplication   & 6.415 \\
Inner product        & 0.02 \\
EC addition         & 0.017 \\
Division        & 0.001 \\
Multiplication  & 0.00055 \\
\hline
\end{tabular}
\caption{Average computation times of cryptographic operations.}
\label{tab:avg-times}
\end{table}

\begin{figure}[t!]
	\centering
	\includegraphics[width=\columnwidth, keepaspectratio]{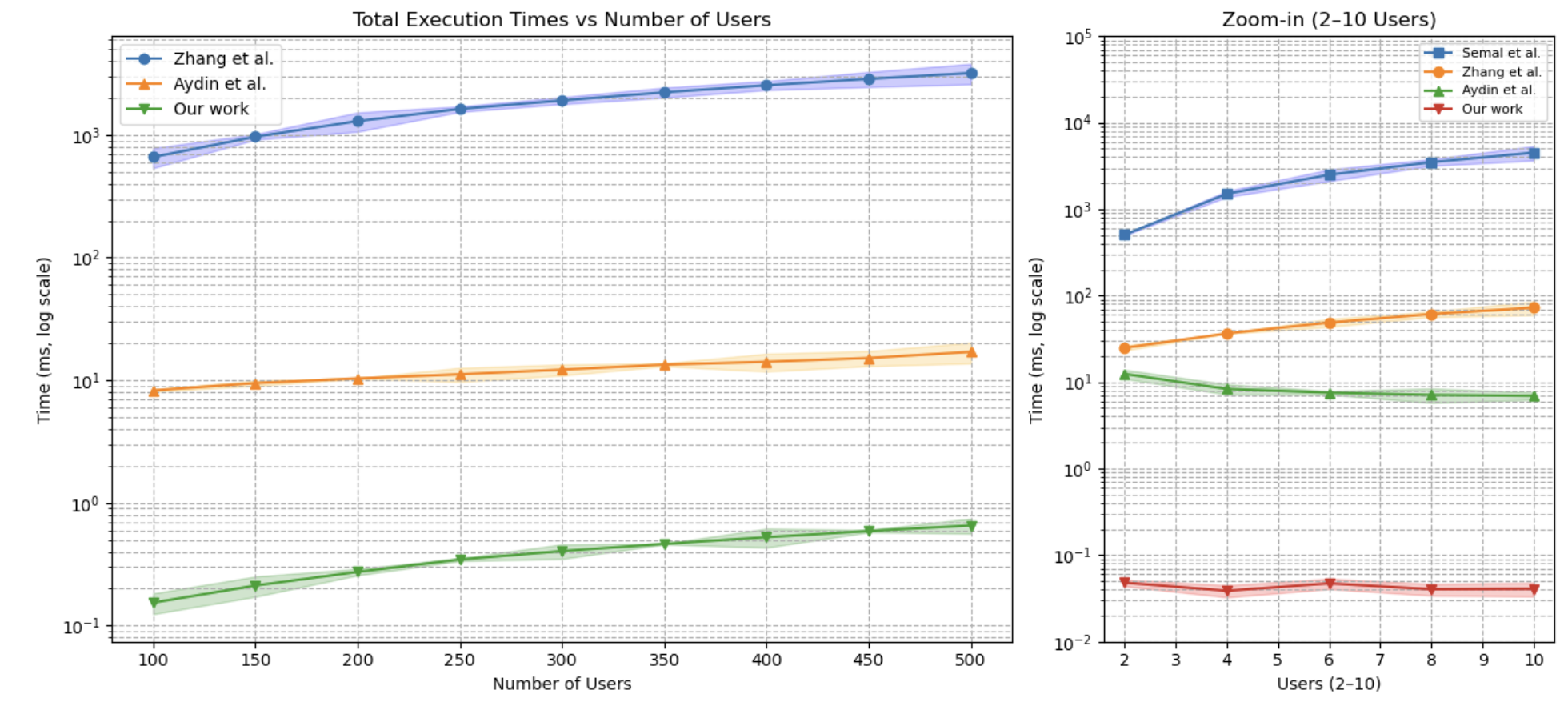}
	\caption{Total performance comparison with respect to the number of users. Since the work of Semal et al. takes a long time to run with a large number of users, the four studies were evaluated in the zoom-in graph for 2–10 users. For fast algorithms, it is normal that the line appears irregular for a small number of users because the differences are very small. The other three studies were evaluated for 100–500 users.}
	\label{fig:total_zoom}
\end{figure}

For a comprehensive understanding, the total authentication process for each group method was analyzed in Figure \ref{fig:total_zoom}, and cumulative time measurements were recorded. Due to the bilinear pairing operation, the work of Semal et al. \cite{semal2018certificateless}is not suitable for large group scales. To better illustrate performance differences, especially at small group sizes, we additionally provide zoom-in views for 2 to 10 users. When we focus on the other three methods, our method is more efficient because it does not involve any costly operations and only relies on multiplication and division operations. Aydin et al. \cite{aydin} and our method are much faster than the other two methods; therefore, at small group sizes their execution times differ only slightly, appearing almost the same. The results presented in Figure \ref{fig:total_zoom} clearly demonstrate that the proposed method surpasses the other methods in terms of total processing time.

\begin{figure}[t!]
	\centering
	\includegraphics[width=\columnwidth, keepaspectratio]{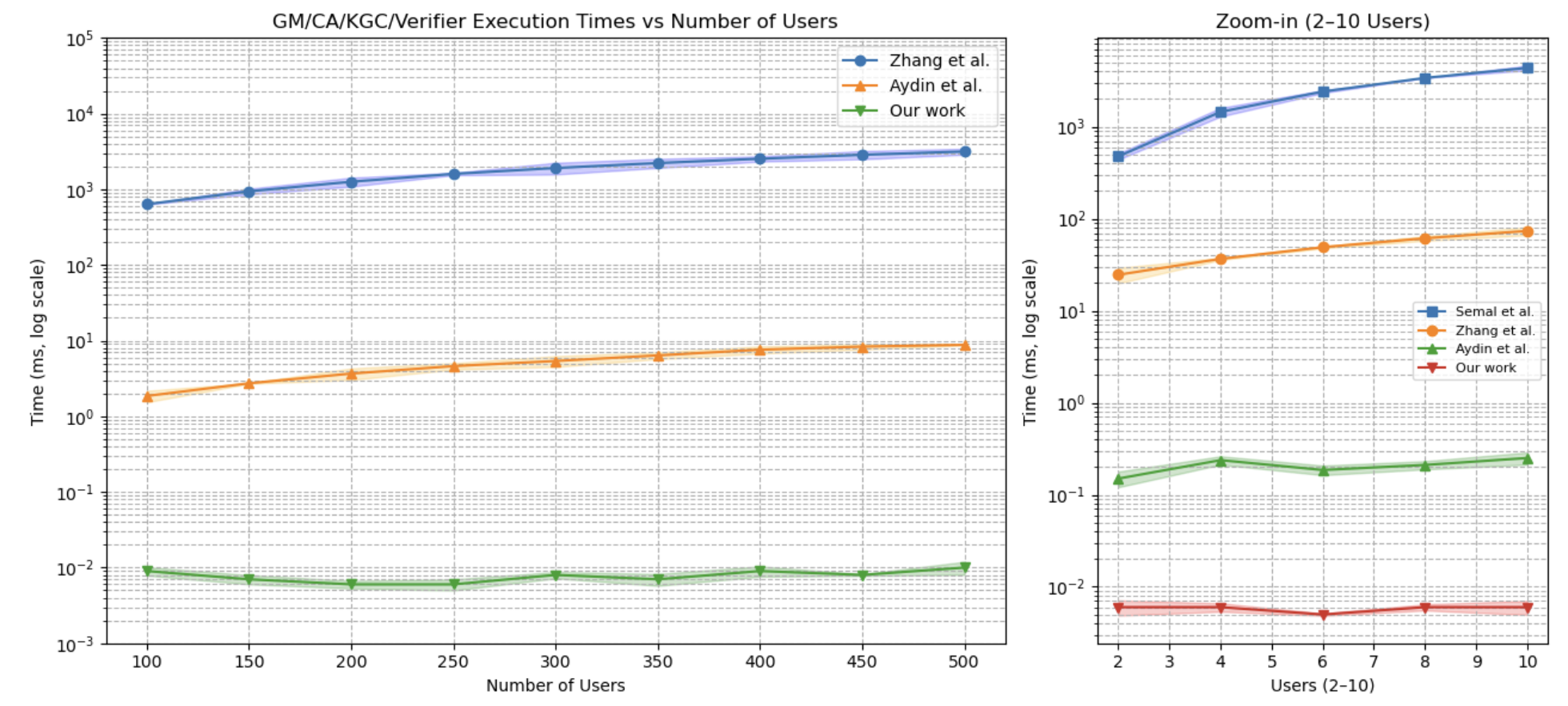}
	\caption{GM/CA/KGC/Verifier performance comparison with respect to the number of users. Since the work of Semal et al. takes a long time to run with a large number of users, the four studies were evaluated in the zoom-in graph for 2–10 users. For fast algorithms, it is normal that the line appears irregular for a small number of users because the differences are very small. The other three studies were evaluated for 100–500 users.}
	\label{fig:GM_zoom}
\end{figure}

The analysis Figure~\ref{fig:GM_zoom} involves identifying the operations performed by GM, CA, Key Generation Center (KGC), and Verifier and implementing them to measure the real-time cost of the scheme for individual users. For the same reason as in Figure~\ref{fig:total_zoom} (Total Execution Times vs. Number of Users), the left side illustrates the results for 100 to 500 users, while the right side provides a zoomed-in view for 2 to 10 users. As shown, Semal et al.\cite{semal2018certificateless} is significantly slower, while Aydin et al. and our method remain lightweight, with our method consistently achieving the lowest execution time.

\begin{figure}[t!]
	\centering
	\includegraphics[width=10cm, height=6.5cm]{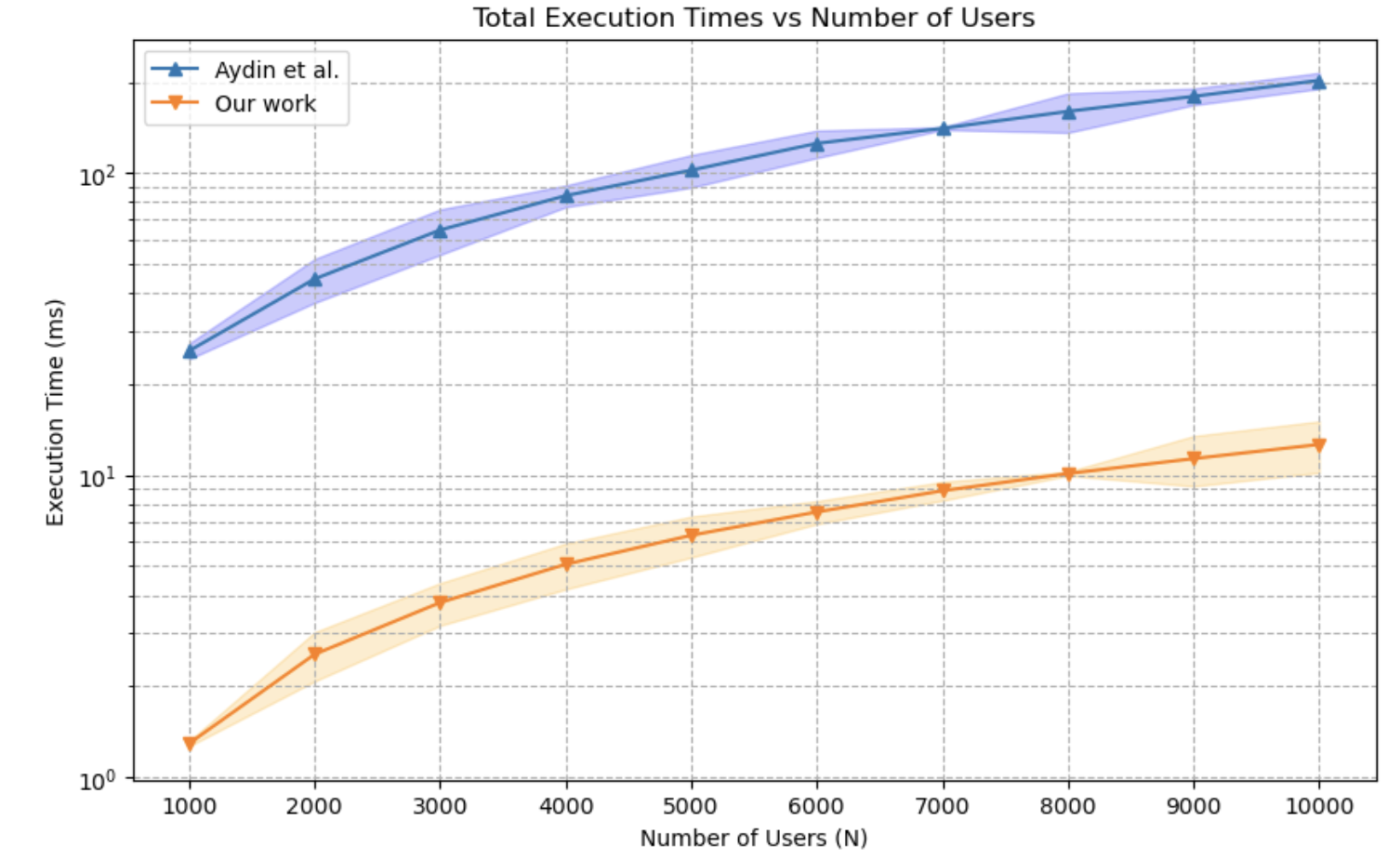}
	\caption{Total performance comparison between Aydin et al. and our proposed scheme. Since they demonstrated good performance, the evaluation was extended to the range of 1,000–10,000 users. For these numbers of users, the performance of our proposed scheme still appears to be satisfactory.}
	\label{fig:Total_scal}
\end{figure}

It is clearly evident that we have two studies Aydin et al. \cite{aydin} and our proposed method, which can be considered lightweight. Then, the Figure \ref{fig:Total_scal} represents the scalability of them as the number of users increases from 1,000 to 10,000. As the number of users increases, the execution time naturally increases since both methods exhibit an exponential trend; however, it should be noted that even with 10,000 users, the performance of our method remains highly efficient.

\newpage

To evaluate the practicality of our scheme under constrained computational environments, we conducted experiments on a Raspberry Pi 4 Model B. This device features 8~GB of RAM and a 1.5~GHz Quad-Core 64-bit ARM Cortex-A72 CPU. Notably, our implementation is strictly single-threaded, utilizing only one core for all operations. This design choice ensures that the performance metrics reflect realistic conditions encountered in IoT-class or embedded devices, which often operate under limited power and processing budgets.

Despite its modest specifications, the Raspberry Pi 4B provides sufficient computational capability to execute the scheme reliably. Its low power consumption approximately 3.0~W in idle state and 6.7~W under full CPU load further reinforces its suitability for energy-sensitive applications.

Figure~\ref{fig:rasp} illustrates the execution times observed when one of the group members operates on such a constrained device. The left-hand graph specifically captures performance metrics under single-core execution, demonstrating that the scheme remains efficient and practical even in low-power environments.

 This experiment was conducted to observe the individual user performance and $GM$ performance under constrained computational resources separately, especially since the workload increases proportionally with the number of users. Even on a resource-constrained device like the Raspberry Pi, the execution time of user remains around 20 ms for up to 1000 users, demonstrating the efficiency and scalability of the proposed scheme. In addition, each user consumed only 0.0032 MB of memory. Also, the execution time for the $GM$ increases exponentially with the number of users; however, it remains below 10 ms even for a group of 1,000 users, demonstrating the scheme’s high efficiency. Memory consumption, on the other hand, was not explicitly measured, as the $GM$ requires only a negligible amount of memory throughout the process. The right-hand graph illustrates the total memory consumption throughout the process, encompassing both the Lagrange interpolation and summation steps. The results show that overall memory usage increases almost linearly with the number of users, with only about a 5 MB difference between 100 and 1,000 users. This indicates that the scheme imposes a lightweight memory load, even as the number of users scales up.

Even when operated on a limited-power device, the execution time takes approximately 20 ms for up to 1,000 users, suggesting that the proposed scheme holds strong potential to satisfy the stringent low-latency demands of next-generation communication technologies, including those beyond 5G and 6G networks.

\begin{figure}[t!]
	\centering
	\includegraphics[width=15cm, height=6.5cm]{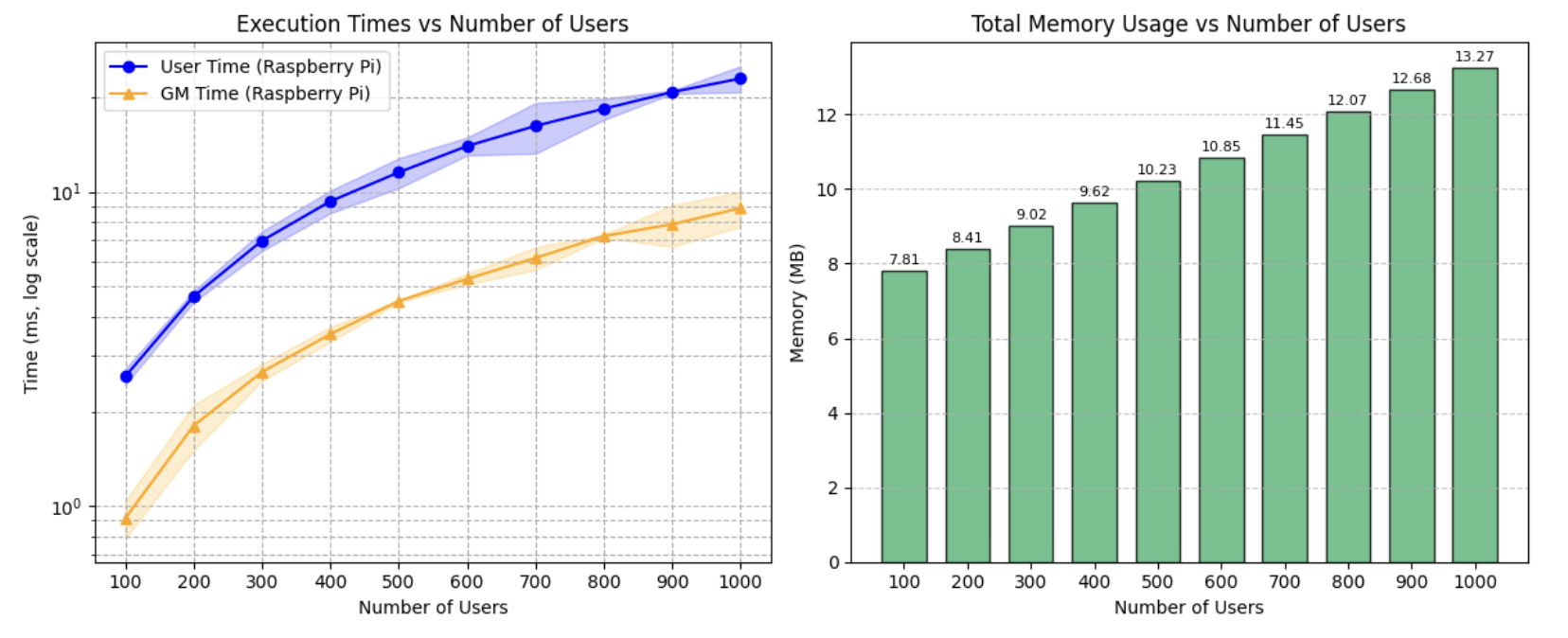}
	\caption{To evaluate the performance of the proposed scheme on limited-power devices, the implementation was tested on a Raspberry Pi. The left-hand graph illustrates the execution time spent by a single user during the authentication phase, while the right-hand graph presents the total memory consumption of the overall process.}
	\label{fig:rasp}
\end{figure}

The findings reinforce the notion that the proposed scheme not only optimizes the user’s individual workload but also contributes to faster overall group authentication, which is essential for scalability and reliability in real-world applications.

\section*{Conclusion}

\subsection*{Contribution}

In the first version of the linear group authentication scheme \cite{Guzey}, a malicious member can give a random basis to any party, allowing that entity to effortlessly join the group communication. The confidentiality of group conversation is provided via a scheme designed for constructing a shared group secret. This secret is then going to be the secret key for the encryption and decryption of the group members' messages. The group secret can easily be obtained by anyone having a basis for the subspace assigned to the group. Therefore, any member easily constructs a random basis by selecting random coefficients for its basis elements and then getting a linear combination of these vectors. Although the first version of the linear group authentication scheme successfully eliminated a major weakness of its predecessors, which were susceptible to DoS attacks, it still has limitations. The scheme can only verify group membership but lacks the capability to identify individual users participating in the communication. Additionally, the algorithm grants each user the ability to add non-members to the group. 

While addressing the vulnerability to DoS attacks has introduced significant flexibility for practical use, it also highlights the need for further improvements in user identification. However, it lacks the ability to identify users participating in the communication. This creates a potential risk where non-members, added to the group without the consent of others, could gain access to all data exchanged among group members. In this study, we enhance the inner product-based group authentication scheme by incorporating additional features that allow both user identification and membership verification. Since authentication and membership confirmation occur simultaneously, this approach is particularly suitable for large-scale environments. The capability to authenticate users joining the process effectively prevents unauthorized access from both within and outside the group.

\subsection*{Future Work}

The proposed work enhances the third-generation group authentication scheme to enable individual authentication. Future wireless communication systems are expected to incorporate space-based entities, implying that the authentication of thousands of users may need to be conducted simultaneously, with authentication results being relayed to others. One of our goal is to enhance the current scheme to handle handover in a practical way. The other one is to remove a central entity to make the method a suitable candidate for autonomous and intelligence systems.

As mentioned in the previous sections, the operation required from the $GM$ is quite lightweight; therefore, a powerful centralized system is not necessary. In this case, decentralized options can also be considered to avoid single-point-of-failure risks. Thus, multiple managers can be employed through blockchain-based group authentication, and instead of relying on a single $GM$ for verification, the process is carried out through consensus with the help of the blockchain.

\section*{Data Availability}
All data generated or analysed during this study are included in this published article.

\bibliography{sample}

\section*{Acknowledgements}

This work was supported by Istanbul Technical University under the Graduate Thesis Project – Doctoral Thesis Project (Project No: MDK-2025-47377).

\section*{Author contributions statement}

The theoretical layout was developed by E.O. and O.G. The implementation and analysis were conducted by O.G. and G.K-K, while the security analysis was handled by E.O. and O.G.

\end{document}